\documentclass[12pt]{article}
\usepackage{amsmath,amsthm,amsfonts,amssymb,epsfig}
\usepackage{color}
\usepackage{subfig}
\usepackage{float}

\newcommand{\reals}{\mathbb{R}}

\newtheorem{stel}{Theorem}
\newtheorem{gevolg}{Corollary}
\newtheorem{lemma}{Lemma}
\theoremstyle{remark}

\begin{document}

\title{Global analysis of a predator-prey model with variable predator search rate}

\author{Ben Dalziel\footnote{Department of Integrative Biology and Department of Mathematics, Oregon State University, Benjamin.Dalziel@oregonstate.edu}, Enrique Thomann\footnote{Department of Mathematics, Oregon State University, homann@math.oregonstate.edu}, Jan Medlock\footnote{Department of Biomedical Sciences, Oregon State University, jan.medlock@oregonstate.edu}
, and Patrick De Leenheer\footnote{Department of Mathematics and Department of Integrative Biology, Oregon State University, Supported in part by NSF-DMS-1411853, deleenhp@math.oregonstate.edu}}

\date{}

\maketitle

\begin{abstract}
We consider a modified Rosenzweig-MacArthur predator-prey model, based on the premise that the search rate of predators is dependent on the prey 
density, rather than constant. A complete analysis of the global behavior of the model is presented, and shows that the model exhibits a dichotomy similar to the classical Rosenzweig-MacArthur model: either the coexistence steady state is  globally stable; or it is unstable, and then a unique, globally stable limit cycle exists. We discuss the similarities, but also important differences between our model and the Rosenweig-MacArthur model. The main differences are that: 1. The paradox of enrichment which always occurs in the Rosenzweig-MacArthur model, does not always occur here, and 2. Even when the 
paradox of enrichment occurs, predators can adapt by lowering their search rate, and effectively stabilize the system.
\end{abstract}



\section{Introduction}
Predator-prey interactions are among the most common in many ecological systems, and have received considerable attention. A prototype model that captures this is:

\begin{eqnarray*}
{\dot N}&=&rN\left(1-\frac{N}{K} \right)-f(N)P\\
{\dot P}&=&P\left(ef(N)-m \right) 
\end{eqnarray*}
Here, $N$ and $P$ denote the prey and predator density respectively, each expressed as numbers per unit area. In the absence of the predator, the prey is assumed to grow logistically, characterized by the positive parameters $r$ and $K$ representing the prey's maximal per capita growth rate, and carrying capacity respectively. The prey is consumed by the predator at a rate $f(N)$ per unit of predator density, and is assumed to depend on the prey density. The choice of the functional form for this rate function $f(N)$ -which is commonly known as the functional response- has important implications for the model behavior. In this paper, we propose a specific functional response that incorporates particular predator behavior, that will be explained below. The positive parameters $m$ and $e$ are the predator's mortality rate, and the conversion efficiency of prey into predator respectively. The parameter $e$ represents the number (or density) of predators obtained, per consumed prey (or density of prey). Since we shall assume throughout this paper that $e$ is a constant, we can scale it out 
by setting ${\bar N}=N$, ${\bar P}=P/e$ and ${\bar f}(N)=ef(N)$. In these transformed variables, and after dropping the bars, the model takes the following form:
\begin{eqnarray}
{\dot N}&=&rN\left(1-\frac{N}{K} \right)-f(N)P\label{s1}\\
{\dot P}&=&P\left(f(N)-m \right) \label{s2}
\end{eqnarray}

A common choice for $f(N)$ is the Holling type II functional response:
\begin{equation}\label{hollingII}
f_{II}(N)=\frac{sN}{shN+1},
\end{equation}
where $s$ and $h$ are positive constants representing the predator's search (or attack) rate, and the handling time respectively. The main qualitative features of this functional are that it is zero when $N$ equals zero, is increasing, saturates for large prey densities at $1/h$, and is concave (the second derivative of $f_{II}(N)$ is negative for all $N\geq 0$). The latter property implies that although the per-predator consumption rate  increases with prey density $N$, it is attenuated (i.e., it slows down) for larger values of $N$.

Using a Holling type II functional response in $(\ref{s1})-(\ref{s2})$  yields the Rosenzweig-MacArthur model \cite{rosenzweig}, which is one of the  benchmark predator-prey models in ecology. To understand the main motivation for this paper, it is useful to review   a mechanistic derivation of the Holling type II functional \cite{gyllenberg,max,smith} here: Consider a sufficiently long window of time $T$ during which an average predator catches 
$M$ prey in a landscape where the prey density is fixed at $N$. Then the functional response equals $M/T$:
$$
f_{II}(N)=\frac{M}{T}.
$$
Let $s$ be the search rate, i.e. the area searched by the average predator per unit of time. If $h$ is the time spent handling a single prey, then the average predator will spend 
a total amount of $T-Mh$ units of time searching for prey, during which the predator covers an area of $s(T-Mh)$. The average predator therefore catches a total of $Ns(T-mh)$ prey, and thus: 
$$
M=Ns(T-Mh).
$$
Dividing by $T$, and solving for $f_{II}(N)=M/T$ yields:
$$
f_{II}(N)=\frac{M}{T}=\frac{sN}{shN+1},
$$
which is Holling's type II functional response. Next we offer a conceptual framework to determine the value of $s$ in practice. Imagine that a predator moves in a plane at a constant velocity, meaning that its direction and magnitude $v$ are fixed. It seems plausible that field biologists can determine relatively accurate estimates of $v$. 
Suppose that at any fixed time, the predator is centered in a disk of radius $r$, and is capable to instantaneously search this disk for prey. Assume now that the predator moves for a period of time $T$ through the plane at the constant velocity $v$. The area searched by the predator in this time interval $[0,T]$ is equal to: $(2r)(vT)+\pi r^2$ (the sum of the area of a rectangle of length $vT$ and width $2r$, and the area of two half-disks with radius $r$). Thus, the search rate during this time interval equals:
$$
2rv+\frac{\pi r^2}{T}.
$$
Letting $T\to +\infty$, we obtain the predator's search rate:
\begin{equation}\label{speed}
s=2vr.
\end{equation}
Clearly, one can make different assumptions on how the predator moves (e.g. by allowing deterministic or random changes to the direction of movement and/or speed $v$; or assume diffusive movement etc), and these will lead to different values of $s$, related to the measurable characteristics of the predator's movement pattern. However, the expression obtained above is obviously a reasonable upper bound of the actual search rate, if we use the predator's largest possible speed, and largest possible radius it can search at any given time, two quantities that are likely well-documented for many predators.

{\bf The main purpose of this paper is to investigate the implications on the model behavior when 
the assumption that the search rate $s$ is constant, is relaxed.}  It seems plausible that when predators 
survey the environment they operate in, and sense the prey density, they may adapt their search rate based on the perceived prey density. 
We shall focus on a case where predators always increase their search rate when they perceive higher prey densities. Moreover, we assume that when the prey is absent, predators cease to search, and that the search rate is limited by a maximally achievable search rate, perhaps due to physiological limitations of the predators and/or physical constraints imposed by the environment. Specifically, we shall consider:
\begin{equation}\label{search}
s(N)=\frac{aN}{N+g},
\end{equation}
where $a$ and $g$ are positive constants. The parameter $a$ is the maximally achievable search rate, and 
$g$ is the half-saturation constant, which corresponds to the prey density at which the search rate is equal to half of the maximal value $a$. 
For all $N>0$, an increase in $g$ leads to a decrease in $s(N)$. In other words, increasing $g$ enables predators to decrease their search rate, a feature with important implications  
that will be discussed later.  Also note that when setting $g=0$ in $(\ref{search})$, we recover a constant search rate, as in the Rosenzweig-MacArthur model. 
We can also easily generalize the conceptual framework used earlier to derive the formula $(\ref{speed})$, to the current context where the search rate is dependent on the prey density $N$. It suffices to assume that the predator makes its speed $v$ dependent on $N$. Specifically, choosing $v(N)=v_{\max}N/(N+g)$ expresses that the predator 
interpolates its speed nonlinearly between zero (when $N=0$), $v_{\max}/2$ (when $N=g$), and $v_{\max}$ (when $N$ becomes infinitely large).  Replacing $v$ by $v(N)$ in $(\ref{speed})$, and $s$ by $s(N)$, yields 
$(\ref{search})$, when we set $a=2v_{\max}r$. This provides us once again with a reasonable way to parameterize the model, and let's us determine the value of $a$ based on 
predator characteristics ($v_{\max}$, $r$ and $g$) that should be readily available in the literature for many predator species.

Starting with Holling's type II functional response $(\ref{hollingII})$, but replacing $s$ by the expression  $s(N)$ in $(\ref{search})$, we obtain the following functional response:
\begin{equation}\label{functional}
f(N)=\frac{aN^2}{ahN^2+N+g}
\end{equation}
The main qualitative features this functional response shares with Holling's type II functional response, is that it is smooth, zero when $N$ equals zero, increasing,
and still saturates at $1/h$ for large prey densities. But the main qualitative difference is that its second derivative changes sign from positive to negative at 
a unique inflection point $N_0$. Consequently, this functional response is an example of what in the literature is known as a Holling type III functional response. 

In this paper, we perform a complete analysis of the global behavior of the model $(\ref{s1})-(\ref{s2})$ when the functional response $f(N)$ is given by $(\ref{functional})$, and compare it to the classical Rosenzweig-MacArthur model obtained when setting $f(N)=f_{II}(N)$ in $(\ref{s1})-(\ref{s2})$. For both models, the most interesting behavior occurs when one assumes that the systems have a steady state where both predator and prey coexist, and when $K>1/ah$ (respectively $K>1/sh$ for the Rosenzweig-MacArthur model).
In this case, both models exhibit a dichotomy: Either the coexistence steady state is globally stable, or it is unstable, and then the systems have a unique globally stable limit cycle. 
But there are fundamental differences between the two models as well.
Indeed, one of the main features of the Rosenzweig-MacArthur model is the so-called Paradox of Enrichment \cite{paradox}. This paradox comes from the observation  that for an increased carrying capacity $K$ for the prey (the 'enrichment' in the paradox), the model can be destabilized, changing its behavior from a system with a globally stable coexistence steady state, to a system with a globally stable limit cycle. This leads to possibly severe fluctuations in both predator and prey that may bring either species close to extinction. For the model presented here, an increase in the carrying capacity $K$ will at first also lead to a similar destabilization phenomenon in some, but interestingly, not in all cases. If the system is destabilized, predators can adaptively lower their search rate (by increasing the model parameter $g$), which in turn lets 
the system regain its pre-existing behavior characterized by the globally stable coexistence steady state. Our results offer an intriguing evolutionary mechanism that  may allow 
predator-prey systems  to cope with the dangers associated to enrichment in the prey's resource.

\section{Preliminaries}
We start by showing that the model $(\ref{s1})-(\ref{s2})$ with $(\ref{functional})$ is well-posed.
\begin{lemma} \label{well-posed}
All solutions of $(\ref{s1})-(\ref{s2})$ with functional response $(\ref{functional})$ remain in the non-negative orthant $\reals^2_+$ when initiated there, exist for all times $t>0$, and remain bounded.
\end{lemma}
\begin{proof}
For all $B>0$, consider the triangular regions
$$
T_B=\{(N,P)\in \reals^2_+  \, | \, N+P\leq B\}.
$$
We claim that $T_B$ is forward invariant for all sufficiently large $B$. This see this, we check that the vector field of the system is inward-pointing on the boundary of each such $T_B$. 
For the boundary parts where $N=0$ or where $P=0$, this is straightforward, where in fact it holds for all $B>0$. To see why it holds when $N+P=B$, note that then
$$
{\dot N}+{\dot P}=rN\left(1-\frac{N}{K} \right)-m(B-N)=-\frac{r}{K}N^2+(r-m)N-mB,
$$
which is negative for all $N\geq 0$, provided that:
$$
(r-m)^2<4\frac{r}{K}mB.
$$
Thus, the vector field is inward-pointing on this part of the boundary of $T_B$, provided that $B$ is sufficiently large.
\end{proof}

{\bf Prey-nullcline}:
For all $N>0$, we define the prey-nulline 
\begin{equation}
P=h(N),
\end{equation}
where
\begin{equation}
 h(N):=\frac{rN\left(1-\frac{N}{K} \right)}{f(N)}=\frac{r}{a} \left(\left(1-\frac{N}{K} \right)(ahN+1)+g\left(\frac{1}{N}-\frac{1}{K} \right) \right).
\end{equation}
It is clear that for fixed $K>0$, the function is smooth for all $N>0$, positive for $0<N<K$, zero at $N=K$, and negative for $N>K$, and that the graph of $h(N)$ has a vertical 
asymptote at $N=0$. We need to understand better the 
qualitative properties of the graph of $h(N)$ on the interval $(0,K]$, which is why we calculate the derivatives of $h(N)$:
\begin{eqnarray}
h'(N)&=&\frac{r}{a}\left(-2\frac{ah}{K}N+ ah-\frac{1}{K} -\frac{g}{N^2} \right)\label{h1}\\
h''(N)&=&\frac{r}{a}\left( -2\frac{ah}{K}+2\frac{g}{N^3}\right)\label{h2}\\
h'''(N)&=&-6\frac{rg}{aN^4}<0,\textrm{ for all }N>0. \label{h3}
\end{eqnarray}
{\bf Case 1}: $K-1/ah\leq 0$. In this case it is clear that $h'<0$ for all $N>0$, and thus  $h(N)$ is decreasing on $(0,K]$.

\noindent
{\bf Case 2}: $K-1/ah>0$. In this case there are two possibilities: Either $h'(N)<0$ for all $N>0$, and then $h(N)$ is decreasing on $(0,K]$ as in {\bf Case 1}. 
Or, there exist $N_{\min}$ and $N_{\max}$ in the interval $(0,(K-1/ah)/2)$, with $N_{\min}\leq N_{\max}$ such that: 
\begin{equation}\label{hprime-signs}
h'(N)=\begin{cases}
<0,\textrm{ if }0<N<N_{\min}\textrm{ and if } N_{\max}<N\leq K\\
0,\textrm{ if } N=N_{\min}\textrm{ and if }N=N_{\max}\\
>0,\textrm{ if } N_{\min} < N < N_{\max}
\end{cases}
\end{equation}
When $N_{\min}=N_{\max}$, then $h(N)$ is decreasing on $(0,K]$. But when $N_{\min}<N_{\max}$, the function $h(N)$ has a unique local minimum at $N=N_{\min}$, and a unique local maximum at $N=N_{\max}$ in the interval $(0,K]$. Furthermore, $h(N)$ is decreasing on $(0,N_{\min})$ and on $(N_{\max},K)$, but increasing on $(N_{\min},N_{\max})$, and  
has a unique inflection point at $N=N_i$, where:
\begin{equation}\label{infl}
N_{\min}<N_i<N_{\max},\textrm{ and }N_i^3=\frac{Kg}{ah},
\end{equation}
and where $h''(N)$ switches from positive to negative when crossing $N=N_i$.

In summary, the function $h(N)$ is either decreasing on $(0,K]$, or it is not. In the latter case, $h(N)$ has exactly two critical points for $N$ in $(0,K]$ (one is a local minimum, the other a local maximum for $h$), and a unique inflection point located between the two critical points. Both possibilities of {\bf Case 2} are illustrated in Figure $\ref{nullcline}$.

\begin{figure}
    \centering
    \subfloat[Decreasing $h(N)$]{{\includegraphics[width=6.5cm]{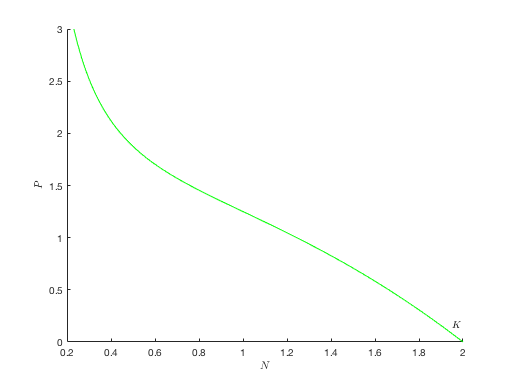} }}
    \qquad
    \subfloat[Non-monotone $h(N)$]{{\includegraphics[width=6.5cm]{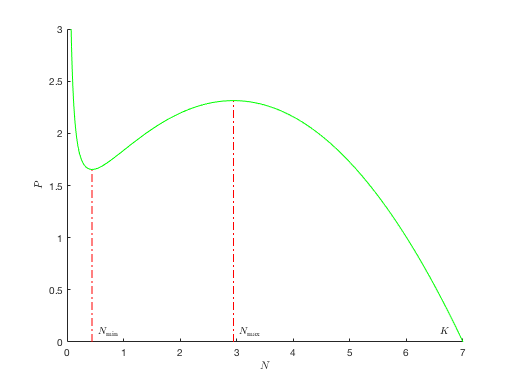} }}
    \caption{Graph of the prey nullcline $P=h(N)$
    of system $(\ref{s1})-(\ref{s2})$ with $(\ref{functional})$. Parameter values: $r=a=h=1$, and (a) $K=2$, and $g=1/2$, (b) $K=7$, and $g=1/7$.}
    \label{nullcline}
\end{figure}





{\bf Predator nullcline}: The predator nullcline is determined by the equation $f(N)=m$, where $f(N)$ is given by $(\ref{functional})$. Since $f$ is increasing, the equation $f(N)=m$ has 
a unique positive solution at $N=N^*$ if and only if $m<1/h$. For convenience we define $N^*=+\infty$ if $m\geq 1/h$. Note that when $N^*< +\infty$, the predator nullcline is given by the vertical line $N=N^*$ in the phase plane $\reals^2_+$ of the system.

From the qualitative behavior of the prey and predator nullclines follows that the model has a unique coexistence steady state $E^*=(N^*,P^*)$ with $P^*=h(N^*)$, if and only if: 
\begin{equation}\label{coexist}
N^*<K.
\end{equation}

We shall first consider the less interesting case when model $(\ref{s1})-(\ref{s2})$ with $(\ref{functional})$ has no coexistence steady state, or equivalently, when $N^*\geq K$. 
The proof is omitted since it is easily obtained using standard phase plane arguments.
\begin{stel}\label{no-coexist}
Assume that $N^*\geq K$. Then system $(\ref{s1})-(\ref{s2})$ with $(\ref{functional})$ has two steady states $E_0=(0,0)$ and $E_1=(K,0)$. All solutions with initial condition 
$(N_0,P_0)$ such that $N_0=0$, converge to $E_0$ at $t\to +\infty$. All solutions with initial condition $(N_0,P_0)$ such that $N_0>0$, converge to $E_1$ as $t\to +\infty$. In particular, the predator always goes extinct.
\end{stel}
This result is not surprising: It says that if the predator's break-even density $N^*$ equals or exceeds the prey's carrying capacity, then the predator is doomed.

Next, we turn to a more interesting scenario where the model has a unique coexistence steady state $E^*=(N^*,P^*)$, but where the prey-nullcline is assumed to decrease on $(0,K]$:
\begin{stel}\label{decreasing-nullcline}
Assume that $N^*<K$, and that $h(N)$ is decreasing for $N$ in $(0,K]$. 
Then system $(\ref{s1})-(\ref{s2})$ with $(\ref{functional})$ has three steady states $E_0=(0,0)$, $E_1=(K,0)$ and a coexistence steady state $E^*=(N^*,P^*)$. All solutions with initial condition 
$(N_0,P_0)$ such that $N_0=0$, converge to $E_0$ at $t\to +\infty$. All solutions with initial condition $(N_0,P_0)$ such that $N_0>0$ and $P_0=0$, converge to $E_1$ as $t\to +\infty$. All solutions with initial condition $(N_0,P_0)$ such that $N_0>0$ and $P_0>0$, converge to $E^*$ as $t\to +\infty$.
\end{stel}
\begin{proof}
The existence of the 3 steady states $E_0$, $E_1$ and $E^*$ is immediate. Linearization of the system yields the following Jacobian matrices at these steady states:
$$
J(E_0)=\begin{pmatrix}
r&0\\
0&-m
\end{pmatrix}, \, 
J(E_1)=\begin{pmatrix}
-r& -f(K)\\
0& f(K)-m
\end{pmatrix},\textrm{ and }
$$
$$
J(E^*)=\begin{pmatrix}
r(1-2N^*/K)-f'(N^*)P^*&-m\\
P^*f'(N^*)&0
\end{pmatrix}=\begin{pmatrix}\frac{m}{r}h'(N^*)&-m\\ P^*f'(N^*)&0 \end{pmatrix},
$$
From this follows that $E_0$ is a saddle, and so is $E_1$ because $N^*<K$ implies that $m=f(N^*)<f(K)$ ($f$ is increasing). Finally, if $h'(N^*)<0$, then 
$E^*$ is a stable because the trace of $J(E^*)$ is negative and its determinant is positive; if $h'(N^*)=0$, then $E^*$ is a center.

The statements regarding the convergence of solutions initiated on the boundary of $\reals^2_+$ are obvious because the boundary is invariant. Thus, to conclude the proof, it suffices to show that every solution initiated in the interior of $\reals^2_+$ converges to $E^*$. Let $\omega(N_0,P_0)$ be the omega limit set of such a solution. 
By Lemma $\ref{well-posed}$, $\omega(N_0,P_0)$ is non-empty. Standard arguments show that neither $E_0$, nor $E_1$ can belong to $\omega(N_0,P_0)$. We shall only prove that $E_0$ cannot belong to $\omega(N_0,P_0)$, because the argument is similar for $E_1$. By contradiction, suppose that $\omega(N_0,P_0)$ contains $E_0$. Then this limit set cannot be equal to the singleton $\{E_0\}$ because that would imply that $(N_0,P_0)$ belongs to the stable manifold of the saddle $E_0$. But this stable manifold coincides with the non-negative $P$-axis, which would contradict that $(N_0,P_0)$ belongs to the interior of $\reals^2_+$.  
Thus, $\omega(N_0,P_0)$ would then also have to contain a point distinct from $E_0$, and then the Butler-McGehee Lemma \cite{smith-waltman} implies that $\omega(N_0,P_0)$ must also contain a point of the stable manifold of $E_0$, distinct from $E_0$. Thus, some point on the positive $P$-axis would be contained in $\omega(N_0,P_0)$, and then forward and backward invariance of omega limit sets would imply that $\omega(N_0,P_0)$ contains the entire positive $P$-axis, contradicting compactness of $\omega(N_0,P_0)$.
To conclude the proof, we must show that $\omega(N_0,P_0)=\{E^*\}$. To do that we invoke the Poincar{\'e}-Bendixson Theorem. If we can establish that the system does not have a nontrivial periodic solution, then the proof will be completed. To rule out periodic solutions, we shall use the Bendixson-Dulac criterion. First we note that any periodic solution must necessarily 
be located in the open strip $S=\{(N,P)\in \reals^2_+\, | \, 0<N< K\textrm{ and } P>0\}$. Indeed, this follows from the fact that the non-negative $N$-axis, and the non-negative $P$-axis are forward invariant for the system, and because $dN/dt\leq -f(N)P<0$ when $N\geq K$ and $P>0$. Next, we multiply the vector field in $(\ref{s1})-(\ref{s2})$ by the function $1/(Pf(N))$ and take the divergence of the scaled vector field, to obtain:
\begin{equation}\label{divergence0}
h'(N)
\end{equation}
Recall that by assumption, $h(N)$ is decreasing for $N$ in $(0,K]$. 
If $h'(N)<0$ for all $N$ in $(0,K]$, then $(\ref{divergence0})$ is negative everywhere in $S$, which concludes the proof in this case. 
A very special case may occur where $h'(N)$ is not negative, but only non-positive for all $N\in (0,K]$. However, in this case $h'(N)$ will have a unique zero 
in this interval. This happens if and only if  
$N_{\min}=N_i=N_{\max}$ (see the earlier discussion of the prey nullcline), and then the zero of $h'(N)$ occurs at this very value. It is clear that in this case, $(\ref{divergence0})$ is still negative almost everywhere in $S$.
\end{proof}

Theorems $\ref{no-coexist}$ and $\ref{decreasing-nullcline}$ leave us with one last case to consider, namely when a unique coexistence steady state $E^*=(N^*,P^*)$ with $N^*<K$ exists, and when the graph of the prey-nullcline $h(N)$ is not decreasing for  $N$ in $(0,K)$, and instead has a local minimum and a local maximum with a unique inflection point sandwiched between the two critical points. The next Section will be devoted to the analysis of this case, but before we proceed, we 
discuss a key property regarding the location of the inflection point $N_i$ of the non-monotone function $h(N)$, and the unique inflection point $N_0$ of the function $f(N)$:
\begin{lemma}\label{inflection}
Assume that $h(N)$ is non-monotone  for $N$ in $(0,K]$, and let $N_i$ be the unique inflection point of $h(N)$ for $N$ in $(0,K]$, and 
$N_0$ be the unique inflection point of $f(N)$ for $N\geq 0$ respectively. Then
$$
N_0<N_i,\textrm{ and thus }
$$
\begin{equation}\label{fdouble}
f''(N)<0\textrm{ for all } N\geq N_i.
\end{equation}
\end{lemma}
\begin{proof}
Let's first locate the inflection point of $f(N)$:
\begin{eqnarray}
f'(N)&=&\frac{aN(N+2g)}{(ahN^2+N+g)^2},\label{f1}\\
f''(N)&=&2a\frac{(N+g)(ahN^2+N+g)-(N^2+2gN)(2ahN+1)}{(ahN^2+N+g)^3} \nonumber\\
&=&2a\frac{-ahN^3-3gahN^2+g^2}{(ahN^2+N+g)^3}=:2a\frac{G(N)}{(ahN^2+N+g)^3}\label{f2}
\end{eqnarray}
Thus, since $G'(N)<0$ for all $N>0$, and $G(0)>0$, there exists a unique $N_0>0$ such that $G(N_0)=f''(N_0)=0$. Moreover, $f''(N)>0 \,(f''(N)< 0) $ for all $N<N_0 \,(N> N_0)$.

By assumption, $h(N)$ is non-monotone in $(0,K]$, hence there exist $N_{\min}$ and $N_{\max}$ in $(0,K]$ with $N_{\min}<N_{\max}$, such that $h'(N_{\min})=h'(N_{\max})=0$. 
Then there exists $N_i$ in $(N_{\min},N_{\max})$ such that $h''(N_i)=0$, and from $(\ref{hprime-signs})$ follows that $h'(N_i)>0$.
From $(\ref{h2})$ we see that $N_i$ is uniquely determined by:
\begin{equation}\label{inflec}
N_i^3=\frac{Kg}{ah},
\end{equation}
Then $h'(N_i)>0$, is equivalent to:
$$
-2\frac{ah}{K}N_i+ah-\frac{1}{K} -\frac{g}{N_i^2}>0\quad \Leftrightarrow \quad ahN_i^2>2\frac{ah}{K}N_i^3 + g + \frac{N_i^2}{K},
$$
and using $(\ref{inflec})$ this implies that:
\begin{equation}\label{key}
ahN_i^2>3g+\frac{N_i^2}{K}.
\end{equation}
Our goal is to show that $N_0<N_i$, or equivalently that $G(N_i)<0$. There holds that:
\begin{eqnarray*}
G(N_i)&=&-ahN_i^3-3gahN_i^2+g^2\\
&=&-ahN_i^3+g(g-3ahN_i^2)\\
&<&-ahN_i^3+g\left(-8g-\frac{3N_i^2}{K}\right),\textrm{ by } (\ref{key})\\
&<&0,
\end{eqnarray*}
which concludes the proof.
\end{proof}

\section{Dichotomy}
We now investigate the most interesting case, which occurs when the system has a unique coexistence steady state $E^*=(N^*,P^*)$, and when the 
prey nullcline $P=h(N)$ is not decreasing for $N$ in $(0,K]$. We have seen before that in this case the prey nullcline has two critical points for $N$ in $(0,K]$, namely a local minimum at $N=N_{\min}$ and a local maximum at $N=N_{\max}$, with a unique inflection point at $N=N_i$, where $N_{\min}<N_i<N_{\max}$.

Recall that the predator nullcline is given by the vertical line $N=N^*$. Depending on the location of $N^*$ in comparison to the critical points $N_{\min}$ and $N_{\max}$ of the prey nullcline $P=h(N)$, we will see that the system displays two distinct dynamical behaviors. When $0<N^*<N_{\min}$, or when $N_{\max}<N^*<K$, the system has a unique, globally stable steady state. This case will be discussed in the next Subsection. When $N_{\min}<N^*<N_{\max}$, the system displays a unique, globally stable limit cycle. This case will be shown in the second Subsection.

\subsection{Globally stable coexistence steady state $E^*$}
Our first main result is as follows:
\begin{stel}\label{main1}
Assume that $N^*<K$, and that $h(N)$ has a local minimum at $N=N_{\min}$ and a local maximum at $N_{\max}$, 
where $0<N_{\min}<N_{\max}<K$. Furthermore, assume that either 
\begin{equation}\label{more}
N_{\max}<N^*,
\end{equation}
or that
\begin{equation}\label{less}
N^*<N_{\min}.
\end{equation}
Then system $(\ref{s1})-(\ref{s2})$ with $(\ref{functional})$ has three steady states $E_0=(0,0)$, $E_1=(K,0)$ and a coexistence steady state $E^*=(N^*,P^*)$. All solutions with initial condition $(N_0,P_0)$ such that $N_0=0$, converge to $E_0$ at $t\to +\infty$. All solutions with initial condition $(N_0,P_0)$ such that $N_0>0$ and $P_0=0$, converge to $E_1$ as $t\to +\infty$. All solutions with initial condition $(N_0,P_0)$ such that $N_0>0$ and $P_0>0$, converge to $E^*$ as $t\to +\infty$.
\end{stel}
\begin{proof}

Before we start the proof, we must introduce some new notation. Recall that we defined $N^*$ as the solution to the equation $f(N)=m$.  However, with that notation, $m$ is asumed to be fixed, but later in this proof we shall need to treat $m$ as a variable parameter. Thus, we redefine $m$ as $m^*$. In other words, in this proof, $N^*$ will denote the unique solution to the equation $f(N)=m^*$.
Now, fixing all parameters $r,K,a,h$ and $g$, but treating $m$ as a variable parameter, the implicit function Theorem implies that $N^*(m)$ (the unique solution of $f(N)=m$) is a smooth map which is increasing on its domain  $(0,1/h)$. It is easy to show that $\lim_{m\to 0+} N^*(m)=0$, and $\lim_{m\to 1/h} N^*(m)= +\infty$, which implies that the map 
$N^*(m)$ is onto $(0,+\infty)$. 

We now turn to the proof of Theorem ${\ref{main1}}$. At first, we can apply the same reasoning as the proof of Theorem $\ref{decreasing-nullcline}$, up to the point where the Bendixson-Dulac criterion is invoked to rule out the existence of 
nontrivial periodic solutions in the open strip $S=\{(N,P)\,|\, 0<N<K, \textrm{ and }P>0\}$. Since $h(N)$ is no longer decreasing for $N$ in $(0,K]$, the scaling function $1/(Pf(N))$ for the vector field used in the proof of Theorem $\ref{decreasing-nullcline}$ is no longer appropriate. Instead, we shall consider a different scaling function here, namely 
$P^{\alpha -1}/f(N)$, where the constant $\alpha$ will be determined later. Scaling the vector field in $(\ref{s1})-(\ref{s2})$ (but where $m$ is replaced by $m^*$, for reasons discussed earlier) by this function, and then taking the divergence, yields:
\begin{equation}\label{divergence}
P^{\alpha} \left(h'(N)+{\alpha} \left(\frac{f(N)-m^*}{f(N)} \right) \right)
\end{equation}
Our goal is to show that there exists some $\alpha$ such that this divergence has fixed sign in the strip $S$.

For all $m$ in $(0,1/h)$, we define the following function for all $N$ in $(0,K]$ with $N\neq N^*(m)$:
$$
\alpha(N,m)=-\frac{f(N)h'(N)}{f(N)-m}
$$
\noindent
{\bf Case 1}: $N_{\max}<N^*$. Recall that $N^*(m)$ is onto $(0,+\infty)$, and thus there exists 
$m_{\max}<m^*$ such that $N^*(m_{\max})=N_{\max}$. We wish to investigate the graph of the function $\alpha(N,m_{\max})$, and claim that:
\begin{enumerate}
\item $\alpha(N,m_{\max})$ is continuous for $N$ in $[0,K]$, and 
$$
\lim_{N\to 0+}\alpha(N,m_{\max})=-\frac{r}{m_{\max}},\textrm{ and }\lim_{N\to N_{\max}}\alpha(N,m_{\max})=:\alpha^*>0.
$$
\item $\alpha(N,m_{\max})$ is increasing on $[0,K]$.
\item For $m^*>m_{\max}$, the graph of $\alpha(N,m_{\max})$ lies above the graph of $\alpha(N,m^*)$ for $N$ in $(N_{\min},N_{\max})$, but below it for $N$ in $(N^*,K]$, see Figure $\ref{alpha1}$.
\end{enumerate}

\begin{figure}
    \centering
        \includegraphics[width=4.0in]{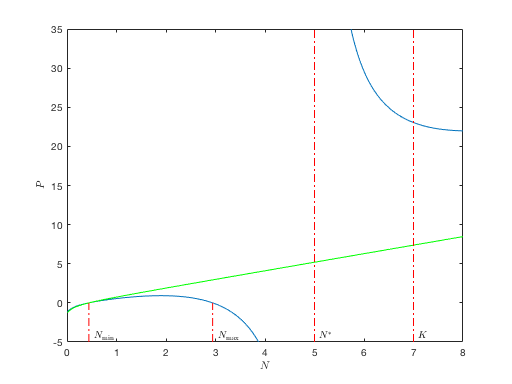}
    \caption{Graph of $\alpha(N,m^*)$ (blue) and $\alpha(N,m_{\max})$ (green). Parameter values: $r=a=h=1$, $K=7$, $g=1/7$, $m^*=0.8294$ (and $N^*=5$), $m_{\max}=0.7373$ (and $N_{\max}=2.9422$).}
    \label{alpha1}  
\end{figure}

Items 1, 2 and 3, together with the fact that $\alpha(N,m^*)\leq  0$ when $N$ belongs to $(0,N_{\min}]$ or to $[N_{\max},N^*)$, and the fact that $\alpha^*$ defined in item 1 above is positive, 
imply that the divergence of the scaled vector field in $(\ref{divergence})$ is negative in the strip $\{(N,P)\, |\, 0<N<K, P>0\}$ when we set $\alpha=\alpha^*$.

{\bf Proofs of the 3 items above}:
\begin{enumerate}
\item We calculate:
\begin{eqnarray*}
&&\lim_{N\to 0+} \alpha(N,m_{\max})\\
&=&\frac{1}{m_{\max}}\lim_{N\to 0+}\frac{aN^2}{ahN^2+N+g}.\; \frac{r}{a}\left(-2\frac{ah}{ K}N+\left(ah-\frac{1}{ K} \right)-\frac{g}{N^2} \right)\\
&=&\frac{r}{m_{\max}}\lim_{N\to 0+}\frac{1}{ahN^2+N+g}.\; \left(-2\frac{ah}{ K}N^3+\left(ah-\frac{1}{ K} \right)N^2-g \right)\\
&=&-\frac{r}{m_{\max}}
\end{eqnarray*}
By de L'Hopital's rule:
\begin{equation*}
\lim_{N\to N_{\max}} \alpha(N,m_{\max})=\lim_{N\to N_{\max}}\frac{-f'h'-fh''}{f'}=\lim_{N\to N^*}\frac{-fh''}{f'}=:\alpha^*>0
\end{equation*}
Continuity of $\alpha(N,m_{\max})$ is now obvious.

\item After simplification, and using the specific expression $(\ref{functional})$ of the functional $f(N)$, 
we have that for $N>0$, the partial derivative of $\alpha$ with respect to $N$ is given by:
\begin{eqnarray}
\alpha'(N,m_{\max})&=& \frac{m_{\max}f'h'-fh''(f-m_{\max})}{(f-m_{\max})^2} \nonumber \\
&=&\frac{r}{(f-m_{\max})^2(ahN^2+N+g)^2}g(N) \label{alpha-prime},
\end{eqnarray}
where
\begin{eqnarray}
g(N)&=&2 \frac{ah}{K}a (1-m_{\max}h)N^4-4m_{\max}\frac{ah}{K}N^3+(-6m_{\max}g\frac{ah}{K}+m_{\max}(ah-\frac{1}{K}))N^2 \nonumber \\
&&+ 2g(m_{\max}(ah-\frac{1}{K})-a(1-m_{\max}h))N+m_{\max}g\label{quartic}
\end{eqnarray}

We wish to show that this function is zero for $N=N_{\max}$, and positive for any $N\neq N_{\max}$ in the interval $(0,K]$. From this, the desired result follows.

First, we claim that $g(N)$ has a zero of at least second order at $N_{\max}$ (i.e. $g(N_{\max})=g'(N_{\max})=0$). To see this, note that it follows from $(\ref{alpha-prime})$ that for $N>0$:
$$
rg(N)=(ahN^2+N+g)^2\left(m_{\max}f'h'-fh''(f-m_{\max})\right),
$$
from which it is clear that $g(N_{\max})=0$. Furthermore, taking the derivative with respect to $N$, yields:
\begin{eqnarray*}
rg'(N)&=&2(ahN^2+N+g)(2ahN+1)\left(m_{\max}f'h'-fh''(f-m_{\max})\right)\\
&&+(ahN^2+N+g)^2\left(m_{\max}f''h'-(f-m_{\max})(2f'h''+fh''')\right),
\end{eqnarray*}
from which also follows that $g'(N_{\max})=0$.

Thus, there exist constants $\alpha,\beta$ and $\gamma$ such that the 4th order polynomial $g(N)$ can be factored as:
$$
g(N)=(N-N_{\max})^2(\alpha N^2 + \beta N +\gamma)
$$
To determine $\alpha$, $\beta$ and $\gamma$, we identify the above expression with $(\ref{quartic})$, which yields:
\begin{eqnarray*}
\alpha &=&2\frac{ah}{ K}a(1-m_{\max}h)\\
\beta &=&4\frac{ah}{K}a(1-m_{\max}h)(N_{\max}-N^*_0)\\
\gamma &=&\frac{m_{\max}g}{(N_{\max})^2}
\end{eqnarray*} 
where $N^*_0$ is the solution to the equation $f(N)=m_{\max}$ but for the case where $g=0$. It is easy to see that $N_{\max}=N^*(m_{\max})>N^*_0$. 
Since $m_{\max}$ belongs to $(0,1/h)$, there follows that $\alpha>0$, and  then also that $\beta>0$. Finally, $\gamma$ is obviously positive as well. Consequently, 
$g(N)>0$ for all positive $N\neq N_{\max}$.

\item
Observe that for all $m>0$, and as long as $N_{\max}\leq N^*(m)$:
\begin{eqnarray*}
\frac{\partial \alpha}{\partial m }(N,m)&=&-\frac{fh'}{(m-f)^2}
\begin{cases}
<0,\textrm{ for }N \textrm{ in } (N_{\min},N_{\max})\\
>0,\textrm{ for }N\textrm{ in }(N^*(m),K]
\end{cases}
\end{eqnarray*}
From this follows the statement made in item 3.

\end{enumerate}

{\bf Case 2}:  If $N^*<N_{\min}$, then there exists $m_{\min}>m^*$ such that $N^*(m_{\min})=N_{\min}$.
This time we investigate the graph of the function $\alpha(N,m_{\min})$. We claim that:
\begin{enumerate}
\item $\alpha(N,m_{\min})$ is continuous for $N$ in $[0,K]$, and 
$$
\lim_{N\to 0+}\alpha(N,m_{\min})=-\frac{r}{m_{\min}},\textrm{ and }\lim_{N\to N_{\min}}\alpha(N,m_{\min})=:\alpha^*<0.
$$
\item $\alpha(N,m_{\min})$ is increasing on $[0,K]$.
\item The graph of $\alpha(N,m_{\min})$ lies above the graph of $\alpha(N,m^*)$ for $N$ in $(0,N^*)$, but below it for $N$ in $(N_{\min},N_{\max})$, see Figure $\ref{alpha2}$.
\end{enumerate}
The proof of these 3 items is entirely analogous to the proof given in {\bf Case 1}, and therefore omitted. 

To conclude the proof in this case, we note that 
items 1, 2 and 3, together with the fact that $\alpha(N,m^*)\geq  0$ when $N$ belongs to $(N^*,N_{\min})$ or to $(N_{\max},K]$, and the fact that $\alpha^*$ defined in item 1 above is negative, imply that the divergence of the scaled vector field in $(\ref{divergence})$ is negative in the strip $\{(N,P)\, |\, 0<N<K, P>0\}$ when we set $\alpha=\alpha^*$.

\begin{figure}
    \centering
        \includegraphics[width=4.0in]{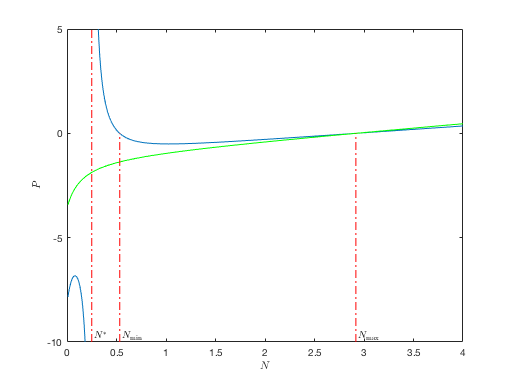}
    \caption{Graph of $\alpha(N,m^*)$ (blue) and $\alpha(N,m_{\max})$ (green). Parameter values: $r=a=h=1$, $K=7$, $g=1/7$, $m^*=0.1220$ (and $N^*=1/4$), $m_{\min}=0.2791$ (and $N_{\min}=0.5326$).}
    \label{alpha2}  
\end{figure}

\end{proof}

\subsection{Unique stable limit cycle}
Our second main result is as follows:
\begin{stel}\label{main2}
Assume that $N^*<K$, and that $h(N)$ has a local minimum at $N=N_{\min}$ and a local maximum at $N_{\max}$, 
where $0<N_{\min}<N_{\max}<K$. Furthermore, assume that 
\begin{equation}\label{between}
N_{\min}<N^*<N_{\max},
\end{equation}
Then the system $(\ref{s1})-(\ref{s2})$ with $(\ref{functional})$ has three steady states $E_0=(0,0)$, $E_1=(K,0)$ and a coexistence steady state $E^*=(N^*,P^*)$. All solutions with initial condition $(N_0,P_0)$ such that $N_0=0$, converge to $E_0$ at $t\to +\infty$. All solutions with initial condition $(N_0,P_0)$ such that $N_0>0$ and $P_0=0$, converge to $E_1$ as $t\to +\infty$. All solutions with initial condition $(N_0,P_0)\neq E^*$ such that $N_0>0$ and $P_0>0$, converge to a unique limit cycle as $t\to +\infty$.
\end{stel}
\begin{proof}
The existence of the 3 steady states $E_0$, $E_1$ and $E^*$ is immediate, and as in the proof of Theorem $\ref{decreasing-nullcline}$, a linearization argument shows that 
$E_0$ and $E_1$ are saddles. The Jacobian matrix at $E^*$ is:
$$
J(E^*)=\begin{pmatrix}
r(1-2N^*/K)-f'(N^*)P^*&-m\\
P^*f'(N^*)&0
\end{pmatrix}=\begin{pmatrix}\frac{m}{r}h'(N^*)&-m\\ P^*f'(N^*)&0 \end{pmatrix},
$$
which has positive trace, from which follows that $E^*$ is unstable. To show the existence of a unique, stable limit cycle, we apply 
Theorem 4.2 in \cite{kuang}. That result is proved under the assumption that $f(N)$ has a simple zero at $N=0$ (i.e. $f(0)=0$, and $f'(0)\neq 0$, see assumption (H3) in \cite{kuang}), which is not satisfied for the 
functional $f(N)$ in $(\ref{functional})$ used here. Indeed, the $f(N)$ used here has a zero of order two at $N=0$ (i.e. $f(0)=f'(0)=0$ and $f''(0)\neq 0$). However, the simplicity of the zero of $f(N)$ 
at $N=0$ is never used in the proof of Theorem 4.2 in \cite{kuang}. Finally, the main condition imposed in \cite{kuang} to establish the existence of a unique periodic solution, is condition (4.18) in that paper. This condition states that a specific function, stated below in $(\ref{unique})$, must be sign-definite for all $N$ in $[0,K]$. However, in our model, this function is only defined for $N$ in $(0,K]$ (this is due precisely to the fact that $f(N)$ has a zero of order two at $N=0$, as  pointed out above). But again, this does not create significant problems. Instead, it suffices to check that this function is sign-definite for $N$ in $(0,K]$, and this will suffice to establish the existence of a unique and stable limit cycle. The condition we need to verify is as follows:
\begin{equation}\label{unique}
mf'h'-f(f-m)h''\geq 0,\textrm{ for all } 0< N\leq K.
\end{equation}
To verify that his condition holds, we shall divide the interval $(0,K]$ into three subintervals, and prove the 
validity of $(\ref{unique})$ on each subinterval. 

\begin{enumerate}
\item $0< N\leq N_i$: 

For fixed parameters $r,K,a,h$ and $g$, and once again treating $m$ as a variable parameter, the implicit function Theorem implies that $N^*(m)$ (the unique solution of $f(N)=m$) is a smooth, and increasing function defined for $m$ in $(0,1/h)$. Recall also that $\lim_{m\to 0+}N^*(m)=0$ and $\lim_{m\to 1/h}N^*(m)=+\infty$, which implies that the map $N^*(m)$ is onto $(0,+\infty)$. Let $m_{\min}<m^*<m_{\max}$ be the 3 values of $m$ where the function $N^*(m)$ equals  $N_{\min}<N^*<N_{\max}$ respectively.

Using the specific expression $(\ref{functional})$ for the functional $f(N)$, the function appearing on the left hand side of the inequality in $(\ref{unique})$ is:
\begin{equation}\label {unique2}
mf'h'-f(f-m)h''=\frac{r}{(ahN^2+N+g)^2}g(N,m),
\end{equation}
where $g(N,m)$ was already defined in $(\ref{quartic})$ (but only for the case that $m=m_{\max}$) as follows:
\begin{eqnarray*}
g(N,m)&=&2 \frac{ah}{K}a (1-mh)N^4-4m\frac{ah}{K}N^3+(-6mg\frac{ah}{K}+m(ah-\frac{1}{K}))N^2 \\
&&+ 2g(m(ah-\frac{1}{K})-a(1-mh))N+mg
\end{eqnarray*}
Our goal is to show that 
\begin{equation}\label{goal}
g(N,m^*)\geq 0, \textrm{ for all } 0 \leq N\leq N_i.
\end{equation}

First, note that $g(N,m)$ is linear in $m$, and recalling $(\ref{h2})$ we can re-write $g(N,m)$ as follows:
\begin{eqnarray}
g(N,m)&=&\left(2\frac{a^2h}{K}N^4-2ag N \right) + ms(N) \nonumber\\
&=&-\frac{a^2}{r}N^4\frac{\partial ^2 h}{\partial N^2}(N)+ms(N) \label{aux1},
\end{eqnarray}
where 
$$
s(N)=-2\frac{(ah)^2}{K}N^4-4\frac{ah}{K}N^3+(-6g\frac{ah}{K} + ah -\frac{1}{K})N^2+2g(2ah-\frac{1}{K})N+g
$$

We have established in item 2 of the proof of both cases of Theorem $\ref{main1}$ that
\begin{equation}\label{aux2}
g(N,m_{\min})\geq 0,\; \textrm{for all } N\geq 0.
\end{equation}

Now, for every $m$ in $(0,1/h)$, and $N\geq 0$, there holds:
$$
\frac{\partial g}{\partial m}=s(N).
$$
Consequently, using $(\ref{aux1})$ and $(\ref{aux2})$, we have that for all $m$ in $(0,1/h)$, and $N$ in $[0, K]$:
\begin{equation}\label{help}
\frac{\partial g}{\partial m}=s(N)=\frac{\partial g}{\partial m}(N,m_{\min})\geq \frac{a^2}{r m_{\min}}N^4\frac{\partial ^2 h}{\partial N^2}(N)
\end{equation}
But $\partial^2 h/\partial N^2(N)\geq 0$ for $N$ in $(0,N_i]$, and thus $(\ref{help})$ and $(\ref{aux2})$ imply that:
$$
g(N,m)\geq 0\textrm{ for all } m\geq m_{\min}\textrm{ and } N \textrm{ in } (0,N_i].
$$
In particular, $(\ref{goal})$ holds.

\item $N_i \leq N \leq N_{\max}$: We distinguish 2 cases, depending on the relative location of $N^*$ and $N_i$:

{\bf Case 1}: $N_i<N^*$. In this case we divide the interval $[N_i, N_{\max}]$ into two further subintervals:

\begin{itemize}
\item $N_i\leq N \leq N^*$: To establish that $(\ref{unique})$ holds when $N$ belongs to this interval, we first evaluate the function in the right-endpoint $N^*$, 
and see that the function is positive there. Next, we calculate the derivative of this function:
$$
mf''h'-(f-m)(2f'h''+fh''')
$$
By inspection it follows that this derivative is negative when $N$ belongs to the interval $[N_i, N^*]$ (here, we have used Lemma $\ref{inflection}$ which implies that $f''(N)<0$ when $N\geq N_i$). Consequently, the function 
$mf'h'-f(f-m)h''$ is decreasing on the interval $[N_i,N^*]$, and as it is positive in the right-endpoint, the function is positive in the entire interval.

\item $N^*\leq N\leq N_{\max}$: It is immediately clear that $(\ref{unique})$ holds when $N$ belongs to the interval $[N^*,N_{\max}]$ by inspection of the signs of the 
various factors and terms in the function $mf'h'-f(f-m)h''$, given the fact that $N_i<N^*\leq N$ when $N$ belongs to this interval, whence $h''(N)\leq 0$. 

\end{itemize}

{\bf Case 2}: $N^*\leq N_i$. In this case, $(\ref{unique})$ is easily seen to hold on the interval $[N_i,N_{\max}]$, using the same argument as in the second item of {\bf Case 1} above.

\item $N_{\max}\leq N\leq K$: We first evaluate the function $mf'h'-f(f-m)h''$ in the left-endpoint $N_{\max}$, and see that the function is positive there. 
The derivative 
$$
mf''h'-(f-m)(2f'h''+fh''')
$$
of this function is positive on the interval $[N_{\max},N^*]$ (here, we have used Lemma $\ref{inflection}$ which implies that $f''(N)<0$ when $N\geq N_{\max}$). 
Consequently,  the function $mf'h'-f(f-m)h''$ is increasing on the interval $[N_{\max},K]$, and as it is positive in the left-endpoint $N_{\max}$, the function is 
positive in the entire interval.
\end{enumerate}
\end{proof}

{\bf Hopf bifurcations are supercritical}: Theorem $\ref{main1}$ and $\ref{main2}$ suggest that when $N^*$ coincides with either $N_{\min}$ (where $h(N)$ achieves a local minimum), or with $N_{\max}$ (where $h(N)$ 
achieves a local maximum), then a Hopf bifurcation occurs. The Jacobian matrix at the coexistence steady state $E^*=(N^*,P^*)$ is:
$$
J(E^*)=\begin{pmatrix}\frac{m}{r}h'(N^*)&-m\\ P^*f'(N^*)&0 \end{pmatrix},
$$
and clearly shows that $E^*$ is a center when $N^*=N_{\min}$ or $N^*=N_{\max}$, and also reveals 
the switch in stability of $E^*$ when $N^*$ crosses either $N_{\min}$ or $N_{\max}$: $E^*$ is a stable spiral when $h'(N^*)<0$, and an unstable spiral when $h'(N^*)>0$. Moreover, using $N^*$ as a bifurcation parameter, it is clear that the eigenvalues of $J(E^*)$ cross the imaginary axis transversally when $N^*$ 
crosses either $N_{\min}$ or $N_{\max}$. Indeed, the  sum of both eigenvalues (which is twice the real part of each eigenvalue) 
equals $(m/r)h'(N^*)$, and the derivative with respect to $N^*$ of this expression is $(m/r)h''(N^*)$, which is positive when $N^*=N_{\min}$, and negative when $N^*=N_{\max}$.
To determine the nature of the Hopf bifurcation (sub- or supercritical), we determine the sign of the following quantity \cite{gail}:
$$
\Omega(N^*)=h''(N^*) \left(2f'(N^*)-\frac{f(N^*)f''(N^*)}{f'(N^*)} \right)+h'''(N^*)f(N^*)
$$
in the cases where $N^*=N_{\min}$, and $N^*=N_{\max}$. When $\Omega(N^*)<0$, the Hopf bifurcation is supercritical, and when $\Omega(N^*)>0$ it is subcritical \cite{gail}. 
We will see that in both cases, $N^*=N_{\min}$ and $N^*=N_{\max}$, the Hopf bifurcation is supercritical. Indeed, suppressing a straightforward 
algebraic calculation using the derivatives $(\ref{h1})$, $(\ref{h2})$ and $(\ref{h3})$ of $h(N)$, and the derivatives $(\ref{f1})$ and $(\ref{f2})$ of $f(N)$, yields that:
\begin{equation*}
\Omega(N^*)=-\frac{2rf(N^*)}{aN^*(N^*+2g)}\left(2\frac{ah}{K}(N^*+3g)+\frac{3g}{(N^*)^2} \right),
\end{equation*}
which is clearly negative when $N^*=N_{\min}$ or $N^*=N_{\max}$. Consequently, we can generalize the conclusion of Theorem $\ref{main1}$, to also include the cases when $N^*=N_{\min}$, and $N^*=N_{\max}$:
\begin{gevolg}\label{all}
Theorem $\ref{main1}$ remains valid if $(\ref{more})$ and $(\ref{less})$ are respectively replaced by 
$$
N_{\max}\leq N^*,\textrm{ and }N^*\leq N_{\min}.
$$
\end{gevolg}

\section{Comparison to the Rosenzweig-MacArthur model}
Here we shall compare the dynamics of the model studied in this paper, to the Rosenzweig-MacArthur model \cite{rosenzweig}, 
obtained by setting $f(N)=f_{II}(N)$ (see $(\ref{hollingII})$) in $(\ref{s1})-(\ref{s2})$. But first we offer some historical perspective. 
Despite the central role of the Rosenzweig-MacArthur model in ecology and mathematical biology, more than 20 years (25 to be precise) has passed between its initial proposal in \cite{rosenzweig}, and a complete and rigorous analysis of its dynamics. The difficulty seems to have been to establish the proof of uniqueness of the limit cycle, which was first announced in \cite{cheng}. According to \cite{gail} however, the proof in \cite{cheng} contained a flaw, which was fixed later in \cite{cheng2}. A concise analysis of the dynamics of the Rosenzweig-MacArthur model can be found in \cite{smith}, and is summarized next. 
First, the Rosenzweig-MacArthur model also always has the extinction steady state $E_0=(0,0)$ and the prey-only steady state $E_1=(K,0)$, just like the model presented here.
The prey nullcline of the Rosenzweig-MacArthur model is a segment of a parabola, given by:
$$
P=\frac{r}{s}\left(1-\frac{N}{K}\right)(shN+1).
$$
The maximum of the parabola is located in the interior of the positive orthant $\reals^2_+$ if and only if:
\begin{equation}\label{up-down}
K>1/sh,
\end{equation}
and in this case this maximum occurs at: 
\begin{equation}\label{max}
{\bar N}_{\max}:=\frac{1}{2}(K-1/sh)
\end{equation}
The predator nullcline is a vertical line $N=N^*$, where $N^*$ is the solution of $f_{II}(N)=m$. Note that $N^*$ exists if and only if $m<1/h$, a condition which is assumed to hold henceforth.   
Therefore, the Rosenzweig-MacArthur model has  a unique coexistence steady state $E^*=(N^*,P^*)$ if and only if $P^*:=(1-N^*/K)(shN^*+1)$ is positive, or equivalently when $N^*<K$. The global dynamics of the Rosenzweig-MacArthur model is summarized next.
\begin{stel} \label{RM}
Consider system $(\ref{s1})-(\ref{s2})$ with $f(N)=f_{II}(N)$ the Holling type II functional response defined in $(\ref{hollingII})$. Assume that $(\ref{up-down})$ 
holds, and that there exists a unique coexistence steady state $E^*=(N^*,P^*)$, in addition to the steady states $E_0=(0,0)$ and $E_1=(K,0)$ which always exist.\\
\noindent
{\bf Case 1}: If ${\bar N}_{\max}\leq N^*$, then $E^*$ is globally asymptotically stable with respect to initial conditions $(N_0,P_0)$ in the interior of $\reals^2$.\\
 \noindent
 {\bf Case 2}: If $N^*<{\bar N}_{\max}$, then $E^*$ is unstable, and there exists a unique stable limit cycle which attracts all solutions with 
initial conditions $(N_0,P_0)\neq E^*$ in the interior of $\reals^2$.
\end{stel}
Comparing this to Corollary $\ref{all}$, we see that the global behavior of the Rosenzweig-MacArthur model exhibits the same dichotomy as the model investigated in this paper: Either the coexistence steady state is globally stable; or it is not, and then a unique, globally stable limit cycle exists. However, a significant difference is that, depending on the location of $N^*$ -the predator's break-even density of prey- there is only a single threshold ${\bar N}_{\max}$ for $N^*$ in the Rosenzweig-MacArthur model that separates 
the two distinct dynamical regimes, and the coexistence steady state is globally stable if and only if ${\bar N}_{\max}\leq N^*$. 
In the model presented here, there are two thresholds $N_{\min}$ and $N_{\max}$ for $N^*$, and the globally stable coexistence steady state occurs when 
$N^*\leq N_{\min}$, or when $N_{\max} \leq N^*$ according to Corollary $\ref{all}$. In other words, here the coexistence steady state is globally stable for all sufficiently large, but also for all sufficiently small values of the predator's break-even density of prey $N^*$, whereas in the Rosenzweig MacArthur model this only happens for all sufficiently large values of $N^*$.

We shall see in a moment that this phenomenon also has important implications in the context of the paradox of enrichment, first pointed out for the Rosenzweig-MacArthur model in \cite{paradox}. Before proceeding to that discussion, we investigate how $N_{\min}$ and $N_{\max}$ in the model studied here, vary with the parameters $K$ and $g$. Recall 
that $N_{\min}$ and $N_{\max}$ are critical points for the function $h(N)$, and thus $h'(N_{\min})=h'(N_{\max})=0$, where $h'(N)$ is given in $(\ref{h1})$. 
\begin{enumerate}
\item {\bf Dependence on $K$}: Fixing all model parameters, except for $K$, and 
assuming that $N_{\min}(K)<N_{\max}(K)$, it follows from implicit differentiation with respect to $K$ of the respective expressions $h'(N_{\min}(K))=0$ and $h'(N_{\max}(K))=0$, 
and using that $h''(N_{\min}(K))>0$ and $h''(N_{\max}(K))<0$, that:
$$
\frac{dN_{\min}}{dK}(K)<0,\textrm{ and }\frac{dN_{\max}}{dK}(K)>0.
$$
Moreover, taking limits for $K\to +\infty$ in $h'(N_{\min}(K))=0$, and in the inequality $N_i=(Kg/ah)^{1/3}<N_{\max}(K)$ -see $(\ref{infl})$- we obtain that:
\begin{equation}\label{K-dep}
\lim_{K\to +\infty} N_{\min}(K)=\left(\frac{g}{ah}\right)^{1/2}=:{\bar N}_{\min},\textrm{ and }\lim_{K\to +\infty} N_{\max}(K)=+\infty .
\end{equation}
These results capture what happens when the prey's carrying capacity $K$ is increased: the gap between the two critical points of the prey nullcline widens, and while  
$N_{\max}(K)$ grows unbounded, $N_{\min}(K)$ is bounded below and converges to a positive value ${\bar N}_{\min}$.
\item {\bf Dependence on $g$}: Fixing all model parameters except for $g$, and assuming that $N_{\min}(g)<N_{\max}(g)$, 
implicit differentiation with respect to $g$ yields in a similar fashion that:
$$
\frac{dN_{\min}}{dg}(g)>0,\textrm{ and }\frac{dN_{\max}}{dg}(g)<0.
$$
Moreover, taking limits for $g\to 0+$ in the inequality $N_{\min}(g)<N_i=(Kg/ah)^{1/3}$ -see $(\ref{infl})$-,  
and in $h'(N_{\max}(g))=0$, we obtain that:
\begin{equation}\label{g-dep}
\lim_{g\to 0+} N_{\min}(g)=0,\textrm{ and }\lim_{g\to 0+} N_{\max}(g)=(K-1/(ah))/2=:{\bar N}_{\max}. 
\end{equation}
In other words, the gap between the critical points of the prey nullcline also grows when $g$ is decreased.
In this case, $N_{\min}(g)$ converges to zero, but $N_{\max}(g)$ is bounded above, and converges to an upper bound ${\bar N}_{\max}$. Note that this bound  is the same as the single threshold defined in $(\ref{max})$ for the Rosenzweig-MacArthur model (when we set $a=s$).
\end{enumerate}
{\bf Paradox of enrichment (or lack thereof)}\\
\noindent
To see why these properties are important in the context of the paradox of enrichment, we first review this paradox for the Rosenzweig-MacArthur model. 
Suppose  that initially, the system parameters are such that $(K-1/(sh))/2={\bar N}_{\max}(K)\leq N^*<K$. By Theorem $\ref{RM}$, the coexistence steady state $E^*$ is globally stable. If all model parameters remain fixed, except for $K$, and if we assume that $K$ is increased to a new value $K_{new}>K$, such that 
$N^*<{\bar N}_{\max}(K_{new})$, then the coexistence steady state is destabilized. The paradox of enrichment is precisely this destabilization phenomenon, 
illustrated in Figure $\ref{par-RM}$.

\begin{figure}
    \centering
    \subfloat[${\bar N}_{\max}(K)<N^*$]{{\includegraphics[width=6.5cm]{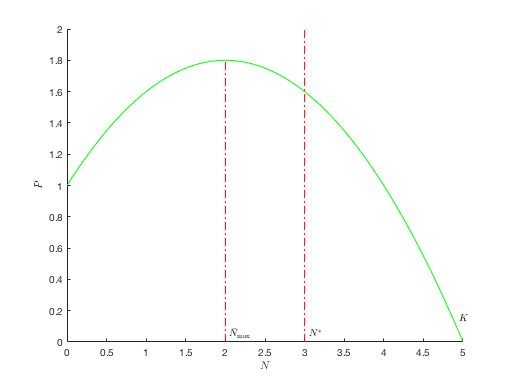}}}
    \qquad
    \subfloat[$N^*<{\bar N}_{\max}(K_{new})$]{{\includegraphics[width=6.5cm]{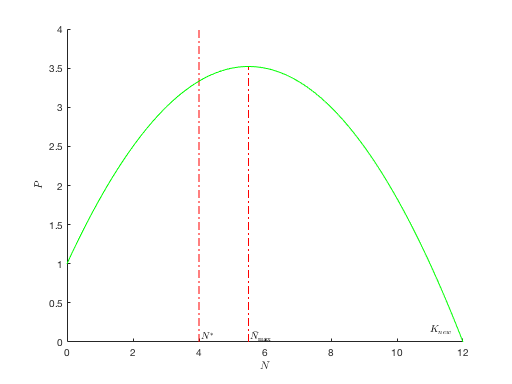}}}
    \caption{Paradox of enrichment in the Rosenzweig-MacArthur model with $r=s=h=1$: (a) $E^*$ is globally stable ($K=5$). (b) $E^*$ is unstable and there is a unique globally stable limit cycle ($K_{new}=12$).}
    \label{par-RM}
\end{figure}

\begin{figure}
    \centering
    \subfloat[$N_{\max}(K)<N^*$]{{\includegraphics[width=6.5cm]{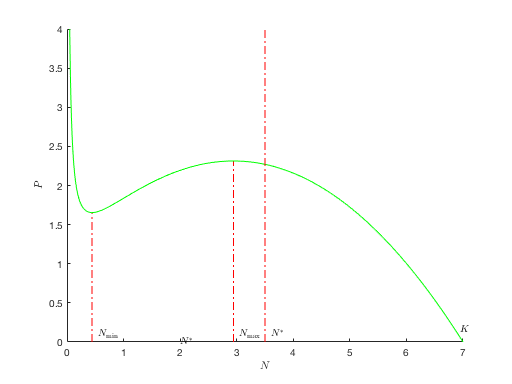}}}
    \qquad
    \subfloat[$N^*<N_{\max}(K_{new})$]{{\includegraphics[width=6.5cm]{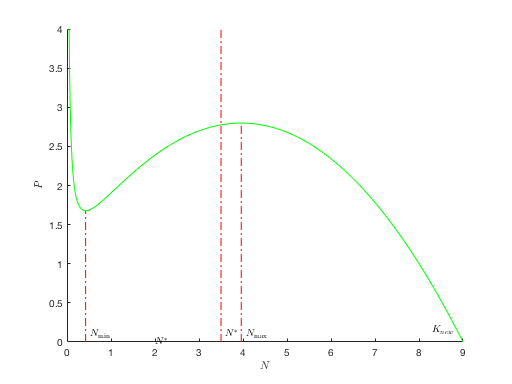}}}
    \caption{Paradox of enrichment in  model $(\ref{s1})-(\ref{s2})$ with $(\ref{functional})$ with parameters $r=a=h=1$, $g=1/7$ and $N^*=3.5$: (a) $E^*$ is globally stable ($K=7$). (b) $E^*$ is unstable and there is a unique globally stable limit cycle ($K_{new}=12$).}
    \label{par-our}
\end{figure}

\begin{figure}
    \centering
    \subfloat[$N^*<{\bar N}_{\min}<N_{\min}(K)$]{{\includegraphics[width=6.5cm]{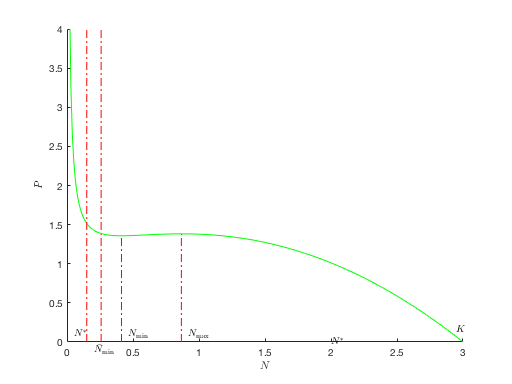}}}
    \qquad
    \subfloat[$N^*<{\bar N}_{\min}<N_{\min}(K_{new})$]{{\includegraphics[width=6.5cm]{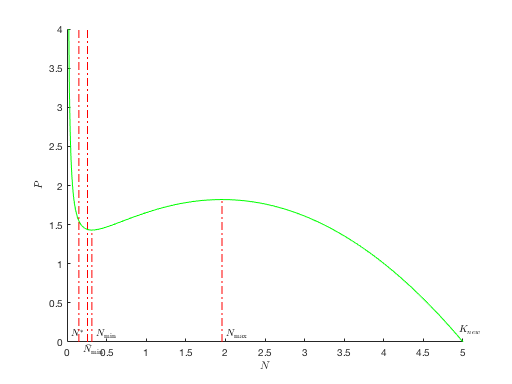}}}
    \caption{No paradox of enrichment in  model $(\ref{s1})-(\ref{s2})$ with $(\ref{functional})$ with parameters $r=a=h=1$, $g=1/15$ and $N^*=0.15$: (a) $E^*$ is globally stable when $K=3$. (b) $E^*$ is globally stable for all $K_{new}>K$ (depicted is $K_{new}=5$).}
    \label{par-our2}
\end{figure}

Let us now investigate whether the paradox of enrichment also occurs for the model presented in this paper. According to Corollary $\ref{all}$, there are two distinct possible initial scenarios that correspond to having a system with a globally stable coexistence steady state: Either $N_{\max}(K)\leq N^*<K$, or $0<N^*\leq N_{\min}(K)$. 
In both cases we shall determine what happens when all model parameters -except for $K$- remain fixed, and when $K$ increases to a new value $K_{new}>K$. 
If initially $N_{\max}(K)\leq N^*<K$, then by $(\ref{K-dep})$ there exist sufficiently large $K_{new}>K$ such that $N_{\min}(K_{new})<N^*<N_{\max}(K_{new})$, which destabilizes the coexistence steady state $E^*$, as illustrated in Figure $\ref{par-our}$. 
Similarly, if initially $0<N^*\leq N_{\min}(K)$, and if also ${\bar N}_{\min}<N^*$, then there exist sufficiently large $K_{new}>K$, such that 
$N_{\min}(K_{new})<N^*<N_{\max}(K_{new})$, once again destabilizing the coexistence steady state $E^*$. However, if initially $0<N^*\leq N_{\min}(K)$, and $N^*\leq {\bar N}_{\min}$, then there are no $K_{new}>K$ that can destabilize $E^*$, as illustrated in Figure $\ref{par-our2}$. . This follows from $(\ref{K-dep})$ because $N^*\leq {\bar N}_{\min}<N_{\min}(K_{new})$, for all $K_{new}>K$. In other words, 
in this last case, the paradox of enrichment does not occur for the model studied here, which is a striking difference with the Rosenzweig-MacArthur model, where the paradox of enrichment always occurs. The role of ${\bar N}_{\min}$, defined in $(\ref{K-dep})$, is that it serves as a buffer: When initially $N^*\leq {\bar N}_{\min}$, the system cannot be destabilized by any enrichment event in the prey's carrying capacity. \\[2ex]
\noindent
{\bf Stabilizing effect when predators decrease their search rate}\\
\noindent
We shall now discuss an important feature of the model studied here that is absent from the Rosenzweig-MacArthur model. Suppose that the system parameters are initially such that the coexistence steady state is unstable, and that a unique globally stable limit cycle exists. This may be the result of an 
enrichment event for the prey's carrying capacity $K$ as described above. Our goal is to show that the predator can respond to this 
by modifying its behavior in a way that stabilizes the coexistence steady state. To achieve this, the predator should simply increase the value of $g$. Recall that this corresponds to a decrease in its  search rate $s(N)$ in $(\ref{search})$, for every $N>0$. To see why this happens, assume that all parameters except for $g$ are fixed, and that $g$ will be increased 
to $g_{new}>g$. 

\begin{figure}
    \centering
    \subfloat[$N_{\min}(g)<N^*<N_{\max}(g)$]{{\includegraphics[width=6.5cm]{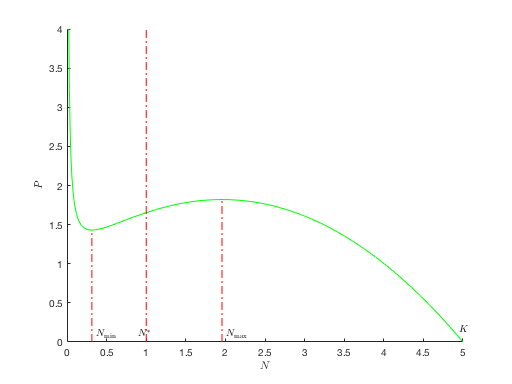}}}
    \qquad
    \subfloat[$h(N)$ is decreasing for $g_{new}$]{{\includegraphics[width=6.5cm]{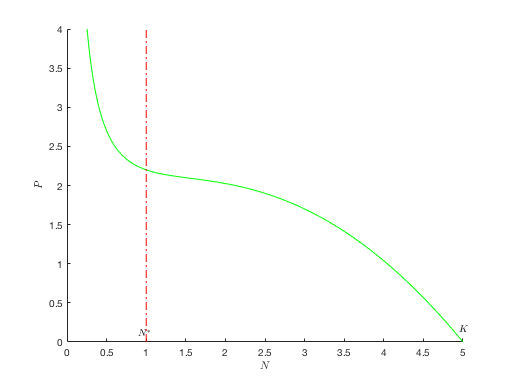}}}
    \caption{Decreased search rate (or increased $g$) stabilizes model $(\ref{s1})-(\ref{s2})$ with $(\ref{functional})$ with parameters $r=a=h=1$, $K=5$, and $N^*=1$: (a) $E^*$ is unstable when $g=1/15$. (b) $E^*$ is globally stable for $g_{new}=0.75$.}
    \label{search-our}
\end{figure}

Thus, we 
assume that initially $N_{\min}(g)<N^*<N_{\max}(g)$, implying that $E^*$ is unstable and that the system has a unique globally stable limit cycle by Theorem $\ref{main2}$. If $g_{new}$ is chosen sufficiently large, then we can ensure that $h'(N)<0$  for all $N$ in $(0,K]$, effectively making the prey nullcline decreasing in $N$, as illustrated in Figure $\ref{search-our}$. 
It follows from Theorem $\ref{main1}$, that in this case $E^*$ is globally stable, which establishes our claim. We can get a better idea of how quickly this happens by considering $(\ref{g-dep})$. By increasing $g$, the gap between $N_{\min}(g)$ and $N_{\max}(g)$ shrinks, and both move towards $N^*$. Global stability of $E^*$ will occur for the first time, when either $N_{\min}(g)$ or 
$N_{\max}(g)$ collides with $N^*$ (by Corollary $\ref{all}$).\\[2ex]
\noindent
{\bf Destabilizing effect (or lack thereof) when predators increase their search rate}\\
To conclude we will demonstrate how an increased predator's search rate $s(N)$, realized by decreasing the parameter $g$, may 
destabilize a globally stable coexistence steady state in certain cases, but not in all cases in the model investigated in this paper. The mechanism turns out to be similar to how the paradox of enrichment following an enrichment event in the prey's carrying capacity can sometimes be avoided, as described above. Suppose that initially, $g$ 
is such that $E^*$ is globally stable. According to  Corollary $\ref{all}$, this means that either  $0<N^*\leq N_{\min}(g)$, or $N_{\max}(g)\leq N^*<K$. 
If $0<N^*\leq N_{\min}(g)$, it follows from $(\ref{g-dep})$, there exist sufficiently small $g_{new}$ such that $N_{\min}(g_{new})<N^*<N_{\max}(g_{new})$, effectively destabilizing $E^*$. 
If $N_{\max}(g)\leq N^*<K$, and if also $N^*<{\bar N}_{\max} $, then there exist sufficiently small $g_{new}$ such that 
$N_{\min}(g_{new})<N^*<N_{\max}(g_{new})<{\bar N}_{\max}$, which again destabilizes $E^*$.
But if $N_{\max}(g)\leq N^*<K$, and if also ${\bar N}_{\max}\leq N^* $, then no matter how small $g_{new}$ is chosen, $(\ref{g-dep})$ implies that 
$N_{\max}(g_{new})<{\bar N}_{\max}\leq N^*$, and then $E^*$ remains globally stable. Thus, whenever ${\bar N}_{\max}\leq N^*$, there are no limits to 
increases in the predator's search rate $s(N)$ that can destabilize the system. The bound ${\bar N}_{\max}$ in $(\ref{g-dep})$ also serves as a buffer for the 
predator's break-even prey density $N^*$, in the sense that if $N^*$ is larger than ${\bar N}_{\max}$, destabilization cannot occur following an increase in the predator's search rate. 

As a final comment, we point out that ${\bar N}_{\max}$ corresponds to the prey density where 
the parabola of the prey nullcline in the Rosenzweig-MacArthur model achieves its maximum (when setting $a=s$). This is not surprising, because taking $g\to 0$ in our model, yields the Rosenzweig-MacArthur model, and when $N^*$ is to the right of this maximum, Theorem $\ref{RM}$ implies that $E^*$ is globally stable.
\section{Conclusions}
Rosenzweig-MacArthur's predator-prey model employs a Holling type II functional response which is predicated on the assumption that the 
predator's search rate is constant, and independent of the prey density. It seems plausible however that predators can modify their search rate, and instead adapt it based on the prey's density. The goal of this paper was to examine the implications on the model behavior when replacing the constant search rate in the Rosenzweig-MacArthur model by a density-dependent search rate $s(N)=aN/(N+g)$, which effectively leads to a Holling type III functional response in the model instead. The following summarizes our findings:
\begin{enumerate}
\item We provided a complete global analysis of the dynamics of the model , showing that just like the Rosenzweig-MacArthur model, the model investigated here exhibits a dichotomy: Either the coexistence steady state is globally stable; or, it is unstable, and then a unique globally stable limit cycle exists (Theorems $\ref{main1}$, $\ref{main2}$ and Corollary $\ref{all}$).
\item Whereas there is a single threshold ${\bar N}_{\max}$ for the predator's break-even prey density $N^*$, that determines which of the two possible regimes occurs in the 
Rosenzweig-MacArthur model, the model presented here can have two thresholds, $N_{\min}<N_{\max}$. 
If the predator's break-even prey density $N^*$ is such that either $N^*\leq N_{\min}$, or if $N_{\max}\leq N^*$, then the  
model has a globally stable coexistence steady state. When $N^*$ is sandwiched between $N_{\min}$ and $N_{\max}$, there is a unique, globally stable limit cycle.
\item 
Whereas the Rosenzweig-MacArthur model always exhibits the paradox of enrichment -a destabilization phenomenon that occurs for all sufficiently strong enrichment events in the prey's carrying capacity $K$- this is not always the case for the model presented here. We identified a threshold ${\bar N}_{\min}=(g/ah)^{1/2}$, such that if $N^*\leq {\bar N}_{\min}$, the model can never be destabilized following an enrichment of the prey's carrying capacity.
\item In those cases where the model studied here, does exhibit destabilization following enrichment in the prey's carrying capacity, the predator can adapt by lowering its search rate, and then the system can always be stabilized again, provided the reduction in the predator's search rate is large enough. 
 This offers an intriguing evolutionary explanation for how predators may have evolved to respond to enrichment events experienced by the prey.
\end{enumerate}

Other mechanisms that can stabilize predator-prey dynamics have been proposed, that rely on certain hypothesized movement patterns of predators and/or prey. Discrete-time, nonlinear host-parasitoid models with aggregation of parasitoids -and where parasitoid aggregation may or may not depend on prey density-  
were investigated in \cite{may} and generalized in \cite{chesson}. A continuous-time, 2-patch predator-prey system with a diffusive predator but static prey was considered in \cite{jansen}. For a more recent review of predator-prey models that incorporate movement of predators and/or prey, as well as spatial heterogeneities in the environment, see \cite{briggs}.
Most of these models are quite complicated due to the fact that explicit decisions have to be made about how the two species move, and because there is a large number of possible scenarios to choose from in this context. Some of these choices are targeted to capture the movement patterns of predators and prey for very specific systems, which may not apply more generally. In contrast, the model presented here neglects explicit spatial effects. Consequently, no decisions on how the two species move have to be made at any stage in the modeling process. Despite the hypothesis of a well-mixed environment, our results indicate that a very simple mechanism -namely, the 
biologically reasonable assumption that predators adapt their search rate based on the perceived prey density- 
always exhibits stabilizing effects on the predator-prey dynamics.

To conclude this paper, we point out that the choice of the search rate $s(N)=aN/(N+g)$ employed here, is very specific. It would be reasonable to ask how robust our conclusions are with respect to changes in this functional $s(N)$. Further research will be needed to answer this question. To caution against unwarranted optimism, we refer to the recent intriguing results in \cite{gail}, where the dynamics of three predator-prey models with distinct functional responses was considered. All three functional responses qualitatively resembled the Holling type II functional response of the Rosenzweig-MacArthur model in the sense that $f(N)$ was assumed to be smooth, zero at $N=0$, increasing but bounded above, and concave (i.e. $f''(N)<0$ for all $N>0$). 
Based on these common features of the functional responses, 
it would be reasonable to expect that these models would  
exhibit the same, or at least similar behavior as the Rosenzweig-MacArthur model. Surprisingly, it was shown in \cite{gail} that this is not the case. One of the models could have two limit cycles, one stable and the other unstable, surrounding a stable coexistence steady state. This implies that this model is bi-stable, with one attractor being a steady state, and another being a stable limit cycle. It is therefore remarkable that the model presented here, which employs a specific example of a Holling type III functional response 
$f(N)$ that transitions from being convex to concave for increasing values of $N$, cannot exhibit more complicated behavior than the original Rosenzweig-MacArthur model.






\newpage
\begin{thebibliography}{199}
\bibitem{smith} Smith, H.L., The Rosenzweig-MacArthur predator-prey model, downloaded from https://math.la.asu.edu/$\sim$halsmith

\bibitem{max} Dawes, J.H.P., and Souza, M.O., 
A derivation of Holling's type I, II and III functional responses in predator-prey systems, Journal of Theoretical Biology 327, p.11-22, 2013. 

\bibitem{gail} Seo, G., and Wolkowicz, G.S.K., 
Sensitivity of the dynamics of the general Rosenzweig-MacArthur model to the mathematical form of the functional response: a bifurcation theory approach, 
Journal of Mathematical Biology 76, p.1873-1906, 2018.

\bibitem{gyllenberg} Geritz, S., and Gyllenberg, M., 
A mechanistic derivation of the DeAngelis-Beddington functional response,
Journal of Theoretical Biology 314, p. 106-108, 2012.

\bibitem{kuang} Kuang, Y., and Freedman, H.I., 
Uniqueness of Limit Cycles in Gause-Type Models of Predator-Prey Systems,
Mathematical Biosciences 88, p. 67-84, 1988.

\bibitem{cheng} Cheng, K.S., 
Uniqueness of a limit cycle for a predator-prey system, 
SIAM Journal on Mathematical Analysis 12, p. 541-548,1981.

\bibitem{cheng2}
Liou, L.P., and Cheng, K.S.  
On the uniqueness of a limit cycle for a predator-prey system, 
SIAM Journal on Mathematical Analysis 19, p. 867-878, 1988.

\bibitem{rosenzweig} Rosenzweig, M.L., and MacArthur, R.H., Graphical representation and stability 
conditions of predator-prey interaction, American Naturalist 97, 209-223, 1963.

\bibitem{paradox}
Rosenzweig, M.L.,  The Paradox of Enrichment, Science 171, p. 385-387, 1971.

\bibitem{smith-waltman}Smith, H.L., and Waltman, P., The theory of the chemostat, Cambridge University Press, 1995.

\bibitem{may}May, R.M., Host-Parasitoid Systems in Patchy Environments: A Phenomenological Model, 
Journal of Animal Ecology 47, p. 833-844, 1978.

\bibitem{chesson} Chesson, P.L., and Murdoch, W.W., Aggregation of Risk-Relationships Among Host-Parasitoid Models, The American Naturalist 127, p. 696-715, 1986.

\bibitem{jansen}Jansen, V.A.A., 
The Dynamics of Two Diffusively Coupled Predator--Prey Populations, 
Theoretical Population Biology 59, p. 199-131, 2001.

\bibitem{briggs}Briggs, C.J., and Hoopes, M.F., Stabilizing effects in spatial  parasitoid-host and predator-prey models: a review. Theoretical 
Population Biology 65, p. 299-315, 2004.


\end{thebibliography}
\end{document}